\documentclass[12pt]{article}
\usepackage{e-jc}

\usepackage{amsmath}
\usepackage{amssymb}
\usepackage{graphicx}

\usepackage{enumerate}

\usepackage{epic}

\newcommand{\C}{\mathcal C}

\usepackage{color}
\newcommand{\blue}{\textcolor{black}}

\dateline{Nov 16, 2023}{XXX}{XXX}


\Copyright{The authors. Released under the CC BY-ND license (International 4.0).}

\title{Phylogenetic trees defined by at most three characters}

\author{Katharina T. Huber\authornote{1}
\and
Simone Linz\authornote{2}
\and
Vincent Moulton\authornote{1}
\and
Charles Semple\authornote{3}
}

\authortext{1}{School of Computing Sciences, University of East Anglia, Norwich NR4 7TJ, United Kingdom (\email{k.huber@uea.ac.uk}, \email{v.moulton@uea.ac.uk}).}
\authortext{2}{School of Computer Science, University of Auckland, Auckland 1142, New Zealand (\email{s.linz@auckland.ac.nz}).}
\authortext{3}{School of Mathematics and Statistics, University of Canterbury, Christchurch 8140, New Zealand (\email{charles.semple@canterbury.ac.nz}).}

\begin{document}

\maketitle

\begin{abstract}
In evolutionary biology, phylogenetic trees are commonly inferred from a set of characters \blue{(partitions)} of a collection of biological entities (e.g., species or individuals in a population). Such characters naturally arise from molecular sequences or morphological data. Interestingly, it has been known for some time that any
binary phylogenetic tree can be (convexly) defined by a set of at most four characters, and that there are binary phylogenetic trees for which three characters are not enough. Thus, it is of interest to characterise those phylogenetic trees that are defined by a set of at most three characters. In this paper, we provide such a characterisation, in particular proving that a binary phylogenetic tree \blue{$T$} is defined by a set of at most three characters precisely \blue{if $T$ has no internal subtree isomorphic to a certain tree}.
\end{abstract}

\section{Introduction}

In evolutionary biology, phylogenetic trees are typically inferred from alignments of molecular sequence data like DNA or protein sequences~\cite{lemey2009phylogenetic}. Each row of such an alignment represents a biological entity (e.g., a species or an individual in a population) 
and each column is referred to as a {\em character}. In mathematical terms, each character is simply a partition of the set of the biological entities in question.
If a character has only two states that, for example, indicate the presence or absence of a biological feature, then 
the character is called binary and corresponds to a bipartition. More frequently, however, biologists analyse data sets that consist of multistate characters, where a character can take on  two or more states.

A fundamental question in the study of character evolution is whether or not a 
collection $\C$ of characters is {\em compatible} \cite[Chapter 4]{semple2003phylogenetics}. 
Biologically speaking, compatibility of $\C$ indicates that there exists a phylogenetic tree $T$ (i.e., an unrooted tree without degree-two vertices whose set $X$ of leaves corresponds to the biological entities) on  which each character $\chi$ in $\C$ evolves without any so-called parallel or reverse transitions. 
This implies that each character state of $\chi$ only evolves once on $T$, in which case
$\C$ is {\it convex} on $T$. Stated another way, for each $\chi$ in $\C$ 
the subtrees of $T$ spanned by the elements in each of the  parts of $\chi$ are pairwise vertex disjoint.

If $\C$ is a collection of binary characters, the Splits Equivalence Theorem~\cite{buneman1971recovery} can be used 
to decide if $\C$ is compatible. Moreover, an elegant graph-theoretic result that is based on chordalisations of the so-called partition intersection graph ${\rm Int}(\C)$ of $\C$ (formally defined in Section~\ref{preliminaries}) characterises when collections of multistate characters 
are compatible~\cite{buneman1974characterisation,meacham1983theoretical,steel1992complexity}. Based on this characterisation, it was further shown in~\cite{sem02} that there exists a certain type of chordalisation of $\rm{Int}(\C)$ that is unique precisely if $\C$ {\it defines} a phylogenetic tree $T$, that is, $\C$ is convex on $T$ and any other phylogenetic tree on which $\C$ is convex \blue{is} isomorphic to $T$.

\begin{figure}
\center
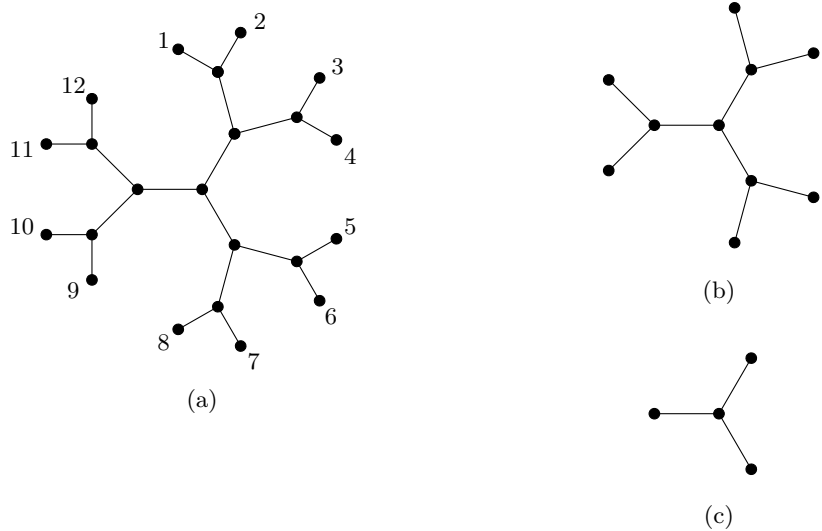
\caption{(a) A phylogenetic tree $T$ with leaf set $X=\{1, 2, \dots, 12\}$ that is not defined by any set of at most three characters. (b) The snowflake. (c) The $3$-star. The snowflake is isomorphic \blue{(in the usual graphical sense)} to the maximal internal subtree of $T$.}
\label{snowflake}
\end{figure}

This last result begs the following question: How many characters are needed to define a given binary phylogenetic tree (i.e., a phylogenetic tree in which
every vertex has degree one or three), when the number of character states is unbounded? Surprisingly, Semple and Steel~\cite{semple2002tree} showed that five characters suffice, a bound that was subsequently sharpened to four by Huber et al.~\cite{huber2005four}. Moreover, as shown in~\cite{semple2002tree}, four is a tight upper bound since there exist binary phylogenetic trees that are not defined by three characters. Indeed, it turns out that the smallest such tree has twelve leaves and is shown in Figure~\ref{snowflake}(a) (see Lemma~\ref{obs:twelve}). Since the collection of binary phylogenetic trees defined by two characters is well understood (see, for example, \cite[Chapter~4.8, Exercise 10]{semple2003phylogenetics} and below), in this paper we provide an answer \blue{to the following problem}: Characterise those binary phylogenetic trees that are defined by a set of at most three characters.

The main result of this paper (Theorem~\ref{thm:main}) gives a solution to this problem in terms of forbidden subtrees in the form of the $6$-leaf tree in
Figure~\ref{snowflake}(b) which is sometimes called the {\em snowflake}. To state it, an {\em internal edge} of a tree $T$ is a non-pendant edge and an {\em internal subtree} of $T$ is a subtree whose edges are all internal.

\begin{theorem}
\label{thm:main}
Let $T$ be a binary phylogenetic tree. Then $T$ is defined by a set of at most three characters if and only if $T$ has no internal subtree isomorphic to the snowflake.
\end{theorem}

The analogous result for binary phylogenetic trees defined by a set of at most two characters is \blue{given by the following theorem}, an immediate consequence of a result \blue{stated in~\cite{semple2003phylogenetics}} (see Theorem~\ref{caterpillar2}). Up to isomorphism, we refer to the unique tree with four vertices, three of which are leaves, as the {\em $3$-star}. \blue{An illustration of the $3$-star is shown in Figure~\ref{snowflake}(c).}

\begin{theorem}
Let $T$ be a binary phylogenetic tree. Then $T$ is defined by a set of at most two characters if and only if $T$ has no internal subtree isomorphic to the $3$-star.
\label{two}
\end{theorem}

Consisting of two main ingredients, the proof of Theorem~\ref{thm:main} essentially works as follows. First, we define three operations. Two of these operations, which we collectively call cherry modifications, extend a binary phylogenetic tree by attaching either one or two new leaves to a cherry, where a cherry refers to two leaves that are adjacent to the same internal vertex. The third operation, which we call a cherry union, amalgamates two binary phylogenetic trees across two cherries with a leaf in common. Second, we analyse sets of three characters that arise from certain edge-colourings of binary phylogenetic trees called internal $3$-colourings. Using these concepts and extending the concept of the partition intersection graph of a set of characters to the partition intersection graph of an internal $3$-colouring, we consider
how the partition intersection graph arising from an internal $3$-colouring behaves relative to the aforementioned operations. In particular, \blue{for proving the necessary direction of Theorem~\ref{thm:main}, we show that if a binary phylogenetic tree $T$ has an internal subtree isomorphic to the snowflake, then, up to isomorphism, $T$ can be obtained from the binary phylogenetic tree shown in Figure~\ref{snowflake}(a) by applying a sequence of cherry modifications, where each modification results in a binary phylogenetic tree not defined by a set of at most three characters (see Theorem~\ref{sequence1}). Conversely, for the sufficient direction of Theorem~\ref{thm:main} (see Theorem~\ref{3-characters}), we inductively show that if a binary phylogenetic tree $T$ does not have an internal subtree isomorphic to the snowflake, then either it is a special type of binary phylogenetic tree or it is the cherry union of two binary phylogenetic trees each of which is defined by a set of at most three characters. In both cases, it will follow that $T$ is defined by a set of at most three characters.}

The rest of this paper is organised as follows. In the next section, we consider sets of characters that define a binary phylogenetic tree, and state some useful results concerning such sets and their relationship with partition intersection graphs from~\cite{semple2002tree} and~\cite{steel1992complexity}. In Section~\ref{sect:colorings}, we show how to define binary phylogenetic trees using internal edge-colourings. In particular, we \blue{essentially} show that any binary phylogenetic tree with at least six leaves is defined by a set of three characters if and only if it can be defined by an internal $3$-colouring (Proposition~\ref{cla:internal}). In Section~\ref{necessary}, we establish the necessary direction of Theorem~\ref{thm:main}. This relies on the two types of cherry modifications. The sufficient direction of Theorem~\ref{thm:main} is established in Section~\ref{sufficient} and relies on the cherry union of two phylogenetic trees. The paper concludes with a brief discussion in Section~\ref{sect:discuss}.

\section{Preliminaries}
\label{preliminaries}

Throughout the paper, $X$ denotes a finite set with $|X|\ge 3$ and, for any positive integer $k$, we set $[k]=\{1, 2, \dots, k\}$. \blue{Furthermore, for a graph $G$, the vertex and edge sets of $G$ are denoted by $V(G)$ and $E(G)$, respectively.} For concepts from phylogenetics, we shall mainly use the terminology given in~\cite{semple2003phylogenetics}.

\begin{figure}
\center
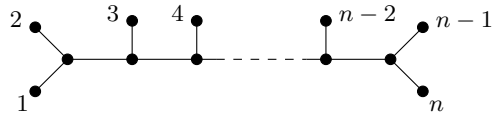
\caption{A caterpillar with leaf set $[n]$, where $n\ge 3$.}
\label{caterpillar1}
\end{figure}

\noindent {\bf Phylogenetic trees.} A {\em phylogenetic \blue{($X$-)}tree} $T$ is a tree with leaf set $X$ and no vertices of degree two. In addition, $T$ is {\em binary} if every \blue{{\em internal vertex}} (i.e., non-leaf vertex) of $T$ has degree three. A binary phylogenetic tree is a {\em caterpillar} if every internal vertex is adjacent to a leaf. Such a phylogenetic tree is shown in Figure~\ref{caterpillar1}. Note that any binary phylogenetic $X$-tree with $3\leq |X|\leq 5$ is a caterpillar.

For a binary phylogenetic $X$-tree $T$, a pair $(x, y)$ of distinct leaves $x, y\in X$ is a {\em cherry} of $T$ if $x$ and $y$ are adjacent to a common vertex. Note that the order of $x$ and $y$ in $(x, y)$ does not matter. Also, note that every binary phylogenetic tree has at least one cherry (see, for example, \cite[Proposition 1.2.5]{semple2003phylogenetics}). In Figure~\ref{example}(a), $(3, 4)$ is a cherry. Furthermore, for a non-empty subset $A\subseteq X$, we let $T(A)$ denote the minimal subtree of $T$ connecting the leaves in $A$.

\begin{figure}
\center
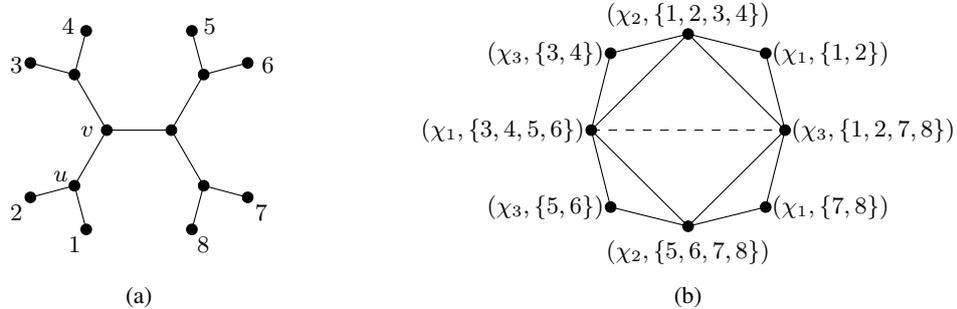
\caption{(a) A binary phylogenetic $X$-tree, where $X=[8]$. (b) The partition intersection graph ${\rm Int}(\C)$ of $\C=\{\chi_1, \chi_2, \chi_3\}$ (solid edges) and a restricted chordal completion of ${\rm Int}(\C)$ (solid and dashed edges), where $\chi_1=\{\{1, 2\}, \{3, 4, 5, 6\}, \{7, 8\}\}$, $\chi_2=\{\{1, 2,3 ,4\}, \{5, 6, 7, 8\}\}$, and $\chi_3=\{\{1, 2, 7, 8\}, \{3, 4\}, \{5, 6\}\}$.}
\label{example}
\end{figure} 

\noindent {\bf Characters.} A {\em character on $X$} is a partition of $X$, that is, a collection of non-empty subsets of $X$ (or {\em parts}) whose pairwise intersections are empty and whose union is $X$. We say that a character $\chi$ on $X$ is {\em convex} on a phylogenetic $X$-tree $T$ if $T(A)$ and $T(B)$ are vertex disjoint for every distinct $A, B\in \chi$, and that a set $\C$ of characters is {\em convex} on $T$ if every character in $\C$ is convex on $T$. Furthermore, $\C$ is {\em compatible} if there is a phylogenetic tree on which $\C$ is convex. A set $\C$ of characters on $X$ {\em defines} $T$ if $\C$ is convex on $T$ and any phylogenetic $X$-tree $T'$ that shares this property with $T$ is \blue{{\em isomorphic}} to $T$ (that is, there is a graph isomorphism between $T$ and $T'$ whose restriction to $X$ is the identity). Note that if $\C$ defines $T$, then $T$ is necessarily binary.

Given a set $\C$ of characters that is convex on a binary phylogenetic tree $T$, we say that $\chi\in \C$ {\em distinguishes} an internal edge $e=\{u, v\}$ of $T$ if there exist distinct $A, B\in \chi$ and distinct elements $x, y\in A$ and $w, z\in B$, such that $u$ but not $v$ lies on the path in $T$ between $x$ and $y$, and $v$ but not $u$ lies on the path in $T$ between $w$ and $z$.
In addition, we say that $T$ is {\em distinguished by $\C$} if every internal edge of $T$ is distinguished by some character in $\C$. To illustrate, consider the collection $\C=\{\chi_1, \chi_2, \chi_3\}$ of characters on $X=[8]$, where $\chi_1=\{\{1, 2\}, \{3, 4, 5, 6\}, \{7, 8\}\}$, $\chi_2=\{\{1, 2, 3, 4\}, \{5, 6, 7, 8\}\}$, and $\chi_3=\{\{1, 2, 7, 8\}, \{3, 4\}, \{5, 6\}\}$. It is easily seen that $\C$ is convex on the binary phylogenetic $X$-tree $T$ shown in Figure~\ref{example}(a). Furthermore, $\C$ distinguishes $T$. For example, the edge $\{u, v\}$ is distinguished by $\chi_1$.

The next lemma is well known but never explicitly stated.

\begin{lemma}
\label{rem:incident}
Let $T$ be a binary phylogenetic $X$-tree, and let $\C$ be a set of characters on $X$ that distinguishes $T$. Then no two incident internal edges of $T$ are distinguished by the same character in $\C$.
\end{lemma}

\begin{proof}
Let $e=\{u, v\}$ and $f=\{v, w\}$ be internal edges of $T$. Since $e$ is distinguished by $\C$, there is a character $\chi$ in $\C$ and states $A$ and $B$ in $\chi$ such that $T(A)$ contains $u$ but not $v$ and $T(B)$ contains $v$ but not $u$. But then $T(B)$ contains $f$ and, in particular $w$, and so $f$ is distinguished by a character in $\C$ that is not $\chi$.
\end{proof}

\noindent {\bf Partition intersection graphs.} Given a set $\C$ of characters, we let ${\rm Int}(\C)$ denote the {\em partition intersection graph of $\C$}, that 
is, the graph with vertex set
$$\{(\chi,A): \mbox{$\chi\in \C$ and $A\in \chi$}\}$$
and edge set
$$\{\{(\chi, A), (\chi', B)\}: A \cap B \neq\emptyset\}.$$
Note that, necessarily, if $(\chi, A)$ and $(\chi', B)$ are joined by an edge, then $\chi\neq \chi'$.

A graph is {\em chordal} if every cycle with at least four vertices has an edge connecting two nonconsecutive vertices. A {\em restricted chordal completion} $G$ of ${\rm Int}(\C)$ is a chordal graph that is obtained from ${\rm Int}(\C)$ by adding only edges that join vertices whose first components are distinct. We refer to the edges of $G$ not in ${\rm Int}(\C)$ as {\em completion edges}. Furthermore, $G$ is {\em minimal} if the deletion of any completion edge of $G$ results in a graph that is not chordal. Continuing the example above, the partition intersection graph ${\rm Int}(\C)$ of $\C=\{\chi_1, \chi_2, \chi_3\}$ is shown in Figure~\ref{example}(b) (solid edges), and a restricted chordal completion of ${\rm Int}(\C)$ is shown in the same figure (solid and dashed edges).

The next two results, \blue{established in}~\cite[Proposition 3]{steel1992complexity} and~\cite[Theorem 1.2]{sem02}, respectively, will be key in what follows. More specifically, they characterise sets of characters that are convex on a phylogenetic tree and sets that define a phylogenetic tree.

\begin{theorem}
\label{thm:key-help1}
Let $\C$ be a set of characters on $X$. Then $\C$ is convex on a phylogenetic $X$-tree if and only if ${\rm Int}(\C)$ has a restricted chordal completion.
\end{theorem}


\begin{theorem}
\label{thm:key-help}
Let $T$ be a binary phylogenetic $X$-tree, and let $\C$ be a set of characters on $X$. Then $\C$ defines $T$ if and only if
\begin{enumerate}[{\rm (i)}]
\item $\C$ is convex on $T$ and $\C$ distinguishes $T$, and
\item ${\rm Int}(\C)$ has a unique minimal restricted chordal completion.
\end{enumerate}
\end{theorem}

\noindent \blue{Because of their frequency of use, Theorems~\ref{thm:key-help1} and~\ref{thm:key-help} will be often used without reference in Sections~\ref{necessary} and~\ref{sufficient}.} 

Theorem~\ref{two} is an immediate consequence of the next theorem. It follows from~\cite[Chapter 4.8, Exercise 10]{semple2003phylogenetics}. However, for completeness, we include a proof.

\begin{theorem}
Let $T$ be a binary phylogenetic tree. Then $T$ is defined by a set of at most two characters if and only if $T$ is a caterpillar.
\label{caterpillar2}
\end{theorem}

\begin{proof}
Let $X$ denote the leaf set of $T$. Using Lemma~\ref{rem:incident}, it is easily checked that $T$ is defined by a set of at most one character if and only if $|X|\in \{3, 4\}$. Thus we may assume that $|X|\ge 5$. If $T$ is defined by a set of two characters, then, by Lemma~\ref{rem:incident}, $T$ has no internal vertex incident with three internal edges. Thus every internal vertex of $T$ is adjacent to a leaf, and so $T$ is a caterpillar.

Conversely, suppose that $T$ is a caterpillar. Without loss of generality, we may assume that the leaf set of $T$ is $[n]$ and that its leaves are labelled as shown in Figure~\ref{caterpillar1}. Say $n$ is even, and consider the set $\{\chi_1, \chi_2\}$ of characters where
$$\chi_1=\{\{1, 2\}, \{3, 4\}, \{5, 6\}, \ldots, \{n-1, n\}\}$$
and
$$\chi_2=\{\{1, 2, 3\}, \{4, 5\}, \{6, 7\}, \ldots, \{n-2, n-1, n\}\}.$$
Now $\{\chi_1, \chi_2\}$ is convex on $T$ and distinguishes $T$. Furthermore, ${\rm Int}(\{\chi_1, \chi_2\})$ is a path, and so ${\rm Int}(\{\chi_1, \chi_2\})$ is chordal. In particular, ${\rm Int}(\{\chi_1, \chi_2\})$ has a unique restricted chordal completion, namely itself. Hence, by Theorem~\ref{thm:key-help}, $\{\chi_1, \chi_2\}$ defines $T$. A similar argument holds if $n$ is odd. Thus if $T$ is a caterpillar, then $T$ is defined by two characters, completing the proof of the \blue{theorem}.
\end{proof}

We end this section with three lemmas. The first is mentioned in~\cite[p.\ 182]{semple2002tree}, and established in~\cite[Section 5]{bordewich2015defining}.

\begin{lemma}
\label{obs:twelve}
The binary phylogenetic tree shown in Figure~\ref{snowflake}(a) is not defined by a set of at most three characters.
\end{lemma}

\begin{lemma}
Let $\C$ be a compatible collection of characters on $X$, and let $G$ be a minimal restricted chordal completion of ${\rm Int}(\C)$. If $e$ is a completion edge of $G$, then $e$ joins two vertices of a vertex-induced cycle of ${\rm Int}(\C)$ with at least four vertices.
\label{completion}
\end{lemma}

\begin{proof}
Let $E'$ be the subset of completion edges of $G$ whose end vertices do not join two vertices in the same vertex-induced cycle of ${\rm Int}(\C)$ of size at least four, and suppose that $E'$ is nonempty. Then, as $G$ is a minimal restricted chordal completion of ${\rm Int}(\C)$, the graph $G\backslash E'$ is not chordal. But then there is a vertex-induced cycle $C'$ of $G\backslash E'$ with at least four vertices all of which are in the same vertex-induced cycle of ${\rm Int}(\C)$. This implies that $C'$ is a vertex-induced cycle of $G$, a contradiction. This completes the proof of the lemma.
\end{proof}

The third lemma shows that a set of characters that defines a binary phylogenetic tree $T$ and contains a character such that one of its parts is a singleton (i.e., has cardinality one) can be slightly modified so that the resulting set of characters still defines $T$, but has one less singleton.

\begin{lemma}
\label{lem:singletons}
Let $\mathcal C$ be a set of characters on $X$, and suppose that $\mathcal C$ defines a binary phylogenetic $X$-tree $T$. Let $\chi$ be a character in $\mathcal C$, and suppose that $A \in \chi$ with $|A|=1$. Then there exists some $B \in \chi - \{A\}$ such that $\mathcal C'=(\mathcal C-\{\chi\})\cup \{\chi'\}$ defines $T$, where
$$\chi'=(\chi-\{A, B\})\cup \{A\cup B\}.$$
\end{lemma}

\begin{proof}
Let $\chi=\{A_1,A_2,\dots,A_k\}$ where $k \ge 2$, and suppose that $A=A_i$ for some $i \in [k]$. Since $\C$ defines $T$, it follows that $\chi$ is convex on $T$. Therefore there exists some $j\in[k]-\{i\}$ and a path in $T$ from the leaf in $A_i$ to a vertex in $T(A_j)$ whose edges are all contained in the set
$$E(T)-\big( E(T(A_1))\cup E(T(A_2))\cup \cdots\cup E(T(A_k))\big).$$
Let $\chi'=(\chi-\{A_i, A_j\})\cup \{A_i\cup A_j\}$ and $\mathcal C'=(\mathcal C-\{\chi\})\cup \{\chi'\}$. By construction, $\chi'$ is convex on $T$, and so $\mathcal C'$ is convex on $T$. We next show that $\mathcal C'$ defines $T$.

If $\mathcal C'$ does not define $T$, then there is a binary phylogenetic $X$-tree $T'$ on which $\mathcal C'$ is convex and $T'$ is not isomorphic to $T$. But then, as $\mathcal C'$ is convex on $T'$, it follows that $\mathcal C$ is also convex on $T'$, contradicting the fact that $\mathcal C$ defines $T$. Thus $\C'$ defines $T$, and so setting $B=A_j$ completes the proof of the lemma.
\end{proof}

\section{Internal $k$-Colourings}
\label{sect:colorings}

Let $k$ be a positive integer, and let $T$ be a phylogenetic $X$-tree. A {\em $k$-assignment}~$\gamma$ is a map $\gamma: E^0(T)\to [k]$, where $E^0(T)$ is
the set of internal edges of $T$. An {\em internal $k$-colouring} of $T$ is a $k$-assignment $\gamma$ with $\gamma(E^0(T))=[k]$ such that every pair of adjacent internal edges in $T$ are assigned different elements in $[k]$. For convenience, we view the elements in $[k]$ as colours. For an internal $k$-colouring $\gamma$ and $c\in [k]$, we let $\pi(\gamma,c)$ denote the character on $X$ that is obtained by removing all internal edges from $T$ that are assigned colour $c$ under $\gamma$ and taking the collection of subsets of $X$ that are contained within each of the resulting connected components. In addition, we set
$$\Pi(\gamma)=\{ \pi(\gamma,c): c\in [k]\}$$
and let ${\rm Int}(\gamma)$ denote the partition intersection graph of~$\Pi(\gamma)$. Clearly, $\Pi(\gamma)$ is convex on $T$ and $T$ is distinguished by $\Pi(\gamma)$. We shall say that a binary phylogenetic $X$-tree is {\em defined by an internal $k$-colouring $\gamma$} if it is defined by $\Pi(\gamma)$. Note that a character on $X$ can contain a singleton whereas a character induced by an internal $k$-colouring, for some $k$, cannot.

The purpose of the next lemma is to clarify under which conditions  a collection~$\C$ of characters that defines a binary phylogenetic tree $T$ equates to a collection $\Pi(\gamma)$ of characters that is induced by an internal $k$-colouring $\gamma$ of $T$.

\begin{lemma}
\label{lem:givescoloring}
Let $T $ be a binary phylogenetic $X$-tree, where $|X|\ge 4$, and let $\C$  be a set of characters on $X$ with cardinality $k$, where \blue{$k\in [3]$}, that is convex on $T$ and also distinguishes $T$. Suppose that no character in $\C$ contains a singleton, that every internal edge in $T$ is distinguished by exactly one character in $\C$,
and that every character in $\C$ distinguishes at least one internal edge of $T$. Then there exists an internal $k$-colouring $\gamma:E^0(T)\to [k]$ such that $\C = \Pi(\gamma)$.
\end{lemma}

\begin{proof}
Let $\C = \{\chi_1,\dots,\chi_k\}$, where \blue{$k\in [3]$}. Consider the $k$-assignment $\gamma: E^0(T) \to [k]$ that takes each internal edge of $T$ to colour $i$ if $\chi_i$ distinguishes that edge. Note that $\gamma$ is well-defined since $|X|\geq 4$ implies that $E^0(T)\not=\emptyset$, and every element of $E^0(T)$ is distinguished by exactly one character in $\C$. Since $\C $ distinguishes $T$, it follows by Lemma~\ref{rem:incident} that any two internal edges in $T$ that are incident with the same vertex are assigned different colours under $\gamma$. Moreover, since every character in $\C$ distinguishes at least one internal edge of $T$, it follows that $\gamma(E^0(T))=[k]$. Hence $\gamma$ is an internal $k$-colouring of $T$. To complete the proof of the lemma, we show that $\C = \Pi(\gamma)$.

For all $i$, let $\pi_i$ denote $\pi(\gamma, c_i)$. Suppose that $\C\neq \Pi(\gamma)$. Since, for all $i$, we have that $\pi_i$ is obtained by deleting all internal edges of $T$ that are distinguished by $\chi_i$, the definition of $\Pi(\gamma)$ implies that there must exist some $\pi_j \in\Pi(\gamma)$ such that \blue{$\chi_j$ refines $\pi_j$} (i.e., there is some $A \in \chi_j$ and some $B \in \pi_j$ such that $A\subsetneq B$). Let $A_1$ and $A_2$ be distinct non-empty subsets of $B$ such that $A_1, A_2\in \chi_j$, and let $P$ be the shortest path in $T$ connecting $T(A_1)$ and $T(A_2)$. Since $|A_1|, |A_2|\ge 2$, the path $P$ consists of internal edges of $T$. Furthermore, as $A_1$ and $A_2$ are subsets of $B$, no edge in $P$ is distinguished by $\chi_j$. To see this, if there is such an edge, then, by construction, $A_1$ and $A_2$ are in different parts of $\pi_j$, \blue{in which case} either $A_1\not\in B$ or $A_2\not\in B$, a contradiction. Thus $P$ has at least two edges; \blue{otherwise, $P$ consists of a single edge distinguished by $\chi_j$}. Let $u$ be the vertex of $P$ that is adjacent to a vertex in $T(A_1)$ but is not in $T(A_1)$, and let $\{u, v\}$ be the edge of $T$ incident with $u$ but not in $P$. If $u$ is incident with three internal edges of $T$, then, as no edge in $P$ is distinguished by $\chi_j$, the edge $\{u, v\}$ is distinguished by $\chi_j$. But then $\chi_j$ is not convex on $T$, a contradiction. It follows that $v$ is a leaf, in which case $v$ appears as a singleton in $\chi_j$; otherwise, the edge of $P$ incident with $u$ and a vertex in $T(A_1)$ is distinguished by $\chi_j$. This last contradiction implies that $\mathcal C=\Pi(\gamma)$.
\end{proof}

Clearly, a binary phylogenetic $X$-tree $T$ is defined by an internal $1$-colouring if and only if $|X|=4$. The next proposition is an immediate consequence of Theorem~\ref{caterpillar2} \blue{and Lemma~\ref{lem:givescoloring}}.

\begin{proposition}
\label{thm:caterpillar}
Let $T$ be a binary phylogenetic $X$-tree, where $|X|\ge 5$. Then $T$ is defined by an internal {\rm 2}-colouring if and only if $T$ is a caterpillar.
\end{proposition}

%

We now turn our attention to internal 3-colourings and establish the main result of this section. In view of Proposition~\ref{thm:caterpillar}, we focus on binary phylogenetic trees that are not caterpillars, that is, binary phylogenetic trees \blue{that have an internal subtree isomorphic to the $3$-star}. The next proposition is used several times in Sections~\ref{necessary} and~\ref{sufficient}.

\begin{proposition}
\label{cla:internal}
Let $T$ be a binary phylogenetic $X$-tree, where $|X|\geq 6$, and suppose that $T$ is not a caterpillar. Then $T$ is defined by a set of three characters on $X$ if and only if $T$ is defined by an internal $3$-colouring of $T$.
\end{proposition}

\begin{proof}
If $T$ is defined by an internal $3$-colouring $\gamma$ of $T$, then $T$ is defined by the set $\Pi(\gamma)$ of characters on $X$ and this set has size three. To prove the converse, that is, if $T$ is defined by a set $\C$ of three characters on $X$, then $T$ is defined by an internal $3$-colouring of $T$, we shall freely use Theorem~\ref{thm:key-help}.
	
Let $\C$ be a set of three characters on $X$ that defines $T$. Note that, by repeated application of Lemma~\ref{lem:singletons}, we may assume that $\mathcal C$ contains no character that contains a singleton.	 Furthermore, since $\mathcal C$ defines $T$, every internal edge of $T$ is distinguished by some character in $\mathcal C$ and, since $T$ is not a caterpillar, every character in $\C$ must distinguish at least one internal edge of $T$. Let $E_2$ denote the set of internal edges of $T$ distinguished by exactly two characters. Observe that, as $|\mathcal C| = 3$ and $T$ has at least two internal edges, no internal edge of $T$ is distinguished by exactly three characters in $\C$.

If $E_2=\emptyset$, then, by Lemma~\ref{lem:givescoloring}, the characters in $\mathcal C$ induce an internal $3$-colouring of $T$, and so the converse holds in this case. Therefore assume that $E_2\not=\emptyset$ and let $e=\{u, v\}$ be an edge in $E_2$. If either $u$ or $v$ is not adjacent to a leaf of $T$, then, as $|\mathcal C|= 3$, we have that $e$ is distinguished by exactly one character in $\mathcal C$, a contradiction. So assume that $u$ and $v$ are adjacent to leaves $x$ and $y$, respectively. We next construct from $\mathcal C$ a set $\mathcal C'$ of three characters that defines $T$ such that $e$ is distinguished by exactly one character in $\mathcal C'$, the number of internal edges distinguished by exactly two characters in $\mathcal C'$ is $|E_2|-1$, and no character in $\mathcal C'$ contains a singleton.
This will complete the proof of the converse since, by repeatedly applying this construction to reduce the number of internal edges distinguished by exactly two characters, we eventually obtain a set $\mathcal C^*$ of three characters on $X$ such that each internal edge of $T$ is distinguished by exactly one character in $\mathcal C^*$, no character  in $\mathcal C^*$ contains a singleton, and each character in $\mathcal C^*$ distinguishes some internal edge of $T$. 

Let $\mathcal C=\{\chi_1, \chi_2, \chi_3\}$, and suppose that $e$ is distinguished by $\chi_1$ and $\chi_2$. Let $e_1$ denote the edge incident with $u$ that is neither $e$ nor $\{u, x\}$, and let $e_2$ denote the edge incident with $v$ that is neither $e$ nor $\{v, y\}$. There are two cases to consider depending on whether (i) either $e_1$ or $e_2$ is pendant, and (ii) neither $e_1$ nor $e_2$ is pendant. In what follows, we prove (ii). The proof for (i) is similar, but more straightforward, and is omitted. Now consider (ii). Since $\mathcal C$ distinguishes every internal edge of $T$, it follows \blue{by Lemma~\ref{rem:incident}} that $e_1$, as well as $e_2$, is distinguished by exactly one character, namely $\chi_3$. Let $A_1$ and $A_2$ be the parts in $\chi_1$ so that $x\in A_1$ and $y\in A_2$. Let $\chi'_1=(\chi_1-\{A_1, A_2\})\cup \{A_1\cup A_2\}$ and let $\mathcal C'=(\mathcal C-\{\chi_1\})\cup \{\chi'_1\}$. We now show that $\mathcal C'$ defines $T$.
	
As $\mathcal C$ is convex on $T$ and distinguishes every internal edge of $T$, it follows that $\mathcal C'$ is also convex on $T$ and distinguishes every internal edge of $T$. Consider the partition intersection graph ${\rm Int}(\mathcal C)$ of $\mathcal C$, and let $G$ be the unique minimal restricted chordal completion of ${\rm Int}(\mathcal C)$. By \cite[Proposition 4.1]{sem02}, ${\rm Int}(\C)$ is connected, and so $G$ is connected. Furthermore $(\chi_3, \{x, y\})$ is a cut vertex of ${\rm Int}(\C)$ and so, by Lemma~\ref{completion}, $(\chi_3,\{x, y\})$ is a cut vertex of $G$.

Now, consider the partition intersection graph ${\rm Int}(\mathcal C')$ of $\mathcal C'$. This graph is obtained from ${\rm Int}(\mathcal C)$ by identifying the vertices $(\chi_1, A_1)$ and $(\chi_1, A_2)$, labelling the identified vertex as $(\chi'_1, A_1\cup A_2)$, deleting one of the two resulting parallel edges joining $(\chi_3, \{x, y\})$ and $(\chi'_1, A_1\cup A_2)$, and relabelling all other vertices of the form $(\chi_1, A_i)$, where $i\not\in \{1, 2\}$ with $(\chi'_1, A_i)$. Observe that, as $(\chi_3, \{x, y\})$ is a cut vertex of ${\rm Int}(\mathcal C)$,
$$\{(\chi_3, \{x, y\}), (\chi'_1, A_1\cup A_2)\}$$
is a vertex-cut of ${\rm Int}(\mathcal C')$. Moreover, as ${\rm Int}(\mathcal C)$ is connected, ${\rm Int}(\mathcal C')$ is connected. 

Let $V_x$ (respectively, $V_y$) be the subset of $V({\rm Int}(\mathcal C))-\{(\chi_3, \{x, y\})\}$ with the property that a vertex $v$ is in $V_x$ (respectively, $V_y$) if 
there is a path from $v$ to $(\chi_1, A_1)$ (respectively, $(\chi_1,A_2)$) in ${\rm Int}(\mathcal C)$ avoiding $(\chi_3, \{x, y\})$. Note that $V_x\not=\emptyset$ because $(\chi_1,A_1)\in V_x$, and that $V_y\not=\emptyset$ because $(\chi_1,A_2)\in V_y$. Similarly, let $V'_x$ (respectively, $V'_y$) denote the subset of $V({\rm Int}(\mathcal C'))$ obtained from $V_x$ (respectively, $V_y$) by deleting $(\chi_1, A_1)$ (respectively, $(\chi_1, A_2)$) and replacing $(\chi_1, A_i)$, where $i\not\in \{1, 2\}$, with $(\chi'_1, A_i)$.

Let $G'$ denote the graph obtained from ${\rm Int}(\mathcal C')$ by joining two vertices in $V({\rm Int}(\mathcal C'))$ with an edge precisely if the corresponding vertices of ${\rm Int}(\mathcal C)$ are joined by a completion edge in $G$. More precisely, $G'$ is constructed from ${\rm Int}(\mathcal C')$ as follows:
\begin{enumerate}[(1)]
\item if vertices of the form $(\chi_2, B)$ and $(\chi_3, C)$ are joined by a completion edge in $G$, then the same vertices are joined by an edge in $G'$;
		
\item if a vertex of the form $(\chi_1, A_i)$ and a vertex $t$ are joined by a completion edge in $G$, where $i\not\in \{1, 2\}$, then $(\chi'_1, A_i)$ and $t$ are joined by an edge in $G'$; and

\item if a vertex of the form $(\chi_1, A_i)$ and a vertex $t$ are joined by a completion edge in $G$, where $i\in \{1, 2\}$, then $(\chi'_1, A_1\cup A_2)$ and $t$ are joined by an edge in $G'$.
\end{enumerate}
Note that, as $(\chi_3, \{x, y\})$ is a cut vertex of $G$, it follows from the construction of $G'$ that 
$$\{(\chi_3, \{x, y\}), (\chi'_1, A_1\cup A_2)\}$$ 
is a vertex-cut of $G'$.
	
We now show that $G'$ is a minimal restricted chordal completion of ${\rm Int}(\mathcal C')$. Assume that $C$ is a vertex-induced cycle of $G'$ with at least four vertices. Then, as $G$ is a restricted chordal completion of ${\rm Int}(\mathcal C)$, and $\{(\chi_3, \{x, y\})\}$ and $\{(\chi_3, \{x, y\}), (\chi'_1, A_1\cup A_2)\}$ are vertex-cuts of $G$ and $G'$, respectively, it follows that $(\chi_3, \{x, y\})$ and $(\chi'_1, A_1\cup A_2)$ as well as a vertex in $V'_x$ and a vertex in $V_y'$ are vertices in $C$. 
But, since $\{(\chi_3, \{x, y\}), (\chi'_1, A_1\cup A_2)\}$ is an edge in \blue{${\rm Int}(\C')$ and thus in} $G$, this is a contradiction as $C$ is a vertex-induced cycle of $G'$ with at least four vertices. It follows that $G'$ is a restricted chordal completion of ${\rm Int}(\mathcal C')$. Furthermore, a similar argument shows that if $G'$ is not a minimal restricted chordal completion of ${\rm Int}(\mathcal C')$, then $G$ is not a minimal restricted chordal completion of ${\rm Int}(\mathcal C)$, a contradiction. Thus $G'$ is a minimal restricted chordal completion of ${\rm Int}(\mathcal C')$.

To see that $G'$ is the unique minimal restricted chordal completion of ${\rm Int}(\C')$, suppose that $G_1'$ is also a minimal restricted chordal 
completion of ${\rm Int}(\mathcal C')$. Since $\{(\chi_3, \{x, y\}), (\chi'_1, A_1\cup A_2)\}$ is a vertex-cut of ${\rm Int}(\mathcal C')$ and the two vertices in this vertex-cut are joined by an edge, it follows by Lemma~\ref{completion} that no completion edge in $G'_1$ joins a vertex in $V'_x$ to a vertex in $V'_y$. Now let $G_1$ be the graph obtained from $G_1'$ by reversing the construction described in (1)--(3) above. That is, $G_1$ is obtained from $G'_1$ as follows:
\begin{enumerate}[(1)$'$]
\item if vertices of the form $(\chi_2, B)$ and $(\chi_3, C)$ are joined by a completion edge in $G'_1$, then the same vertices are joined by an edge in $G_1$;

\item if a vertex of the form $(\chi'_1, A_i)$ and a vertex $t$ are joined by a completion edge in $G'_1$, where $i\not\in \{1, 2\}$, then $(\chi_1, A_i)$ and $t$ are joined by an edge in $G_1$; and

\item if $(\chi'_1, A_1\cup A_2)$ and a vertex $t$ in $V'_x$ (resp.\ $V'_y)$ are joined by a completion edge in $G'_1$, then $(\chi_1, A_1)$ (resp.\ $(\chi_1, A_2)$) and $t$ are joined by an edge in $G_1$.
\end{enumerate}
Arguing as before, it is easily checked that $G_1$ is a minimal restricted chordal completion of ${\rm Int}(\mathcal C)$. By the uniqueness of $G$, it follows that $G_1$ is isomorphic to $G$. Since the constructions given by (1)--(3) and (1)$'$--(3)$'$ undo each other, this implies that $G'_1$ is isomorphic to $G'$. It follows that $G'$ is the unique minimal restricted chordal completion of ${\rm Int}(\mathcal C')$. Hence $\mathcal C'$ defines $T$ and $\mathcal C'$ has the desired properties, in particular, the number of internal edges distinguished by exactly two characters in $\mathcal C'$ is $|E_2|-1$ and no character in $\mathcal C'$ contains a singleton.
\end{proof}


\section{Proof of Theorem~\ref{thm:main}: Necessary Direction}
\label{necessary}

We begin the necessary direction of the proof of Theorem~\ref{thm:main} by describing two operations that extend a binary phylogenetic tree. Let $T$ be a binary phylogenetic $X$-tree, where $|X|\ge 3$, and let $(x, y)$ be a cherry of $T$. The first operation, called a {\em fork modification}, adds a new leaf $z\not\in X$ to $(x, y)$ as shown in Figure~\ref{modifications}(a) to obtain a binary phylogenetic $(X\cup \{z\})$-tree $T'$. Formally, $T'$ is obtained from $T$ by subdividing the pendant edge incident with either $x$ or $y$, say $x$, and then adjoining a new leaf $z\not\in X$ to $T$ by adding an edge joining $z$ and the subdivision vertex. The second operation, called a {\em balanced modification}, adds two new leaves $w, z\not\in X$ to $(x, y)$ as shown in Figure~\ref{modifications}(b) to obtain a binary phylogenetic $(X\cup \{w, z\})$-tree $T''$. More precisely, $T''$ is obtained from $T$ by subdividing each of the pendant edges incident with $x$ and $y$, and then adjoining new leaves $w, z\not\in X$ to $T$ by adding an edge joining one of the leaves, say $z$, to the subdivision vertex on the pendant edge incident with $x$ and an edge joining $w$ to the subdivision vertex on the pendant edge incident with $y$. Collectively, we refer to fork and balanced modifications as {\em cherry modifications}. Although not explicitly needed for the paper, it is straightforward to show that if $T$ is a binary phylogenetic $X$-tree, where $|X|\ge 3$, then, up to isomorphism, $T$ can be obtained from a binary phylogenetic tree on three leaves by a sequence of cherry modifications.

\begin{figure}
\center
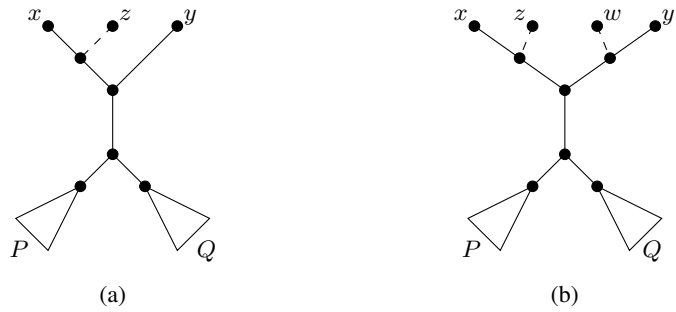
\caption{An illustration of (a) a fork modification and (b) a balanced modification, where $P$ and $Q$ denote the leaf sets of the corresponding subtrees.}
\label{modifications}
\end{figure}

The next lemma shows that a fork modification preserves the property of being defined by at most three characters. For ease of reading, if $\gamma$ is an internal $k$-colouring of a binary phylogenetic tree, we write the vertices of ${\rm Int}(\gamma)$ in the form $(c, A)$ instead of $(\pi(\gamma, c), A)$.

\begin{lemma}
Let $T$ be a binary phylogenetic $X$-tree, where $|X|\ge 3$, and let $T'$ be a binary phylogenetic tree obtained from $T$ by a fork modification. Then $T$ is defined by a set of at most three characters if and only if $T'$ is defined by a set of at most three characters.
\label{fork}
\end{lemma}

\begin{proof}
Let $(x, y)$ be a cherry of $T$. We may assume that $T'$ is obtained from $T$ as shown in Figure~\ref{modifications}(a). Thus $(x, z)$ is a cherry of $T'$ and the leaf set of $T'$ is $X\cup \{z\}$. By Theorem~\ref{caterpillar2}, the lemma holds if $T$ is a caterpillar as $T'$ is also a caterpillar. So we may assume that $T$ has an internal vertex incident with three internal edges and $|X|\ge 6$. There are two cases to consider depending on whether (i) exactly one of $P$ and $Q$ has size one and (ii) $|P|, |Q|\ge 2$. We will establish the lemma for (ii). The proof for (i) is similar and omitted.

Suppose that $|P|, |Q|\ge 2$, and $T$ is defined by three characters. Then, by Proposition~\ref{cla:internal}, $T$ is defined by an internal $3$-colouring $\gamma$. Let $\{c_1, c_2, c_3\}$ be the codomain of $\gamma$ and let $\gamma'$ be an internal $3$-colouring of $T'$ that extends $\gamma$, that is, $\gamma'$ has the property that if \blue{$e\in E^0(T)$}, then $\gamma'(e)=\gamma(e)$. Since $|P|, |Q|\ge 2$, we may assume without loss of generality that
$$(c_1, \{x, y\}),\, (c_2, \{x, y\}\cup P'),\, (c_3, \{x, y\}\cup Q')$$
are vertices of ${\rm Int}(\gamma)$, where $P'$ and $Q'$ are non-empty subsets of $P$ and $Q$, respectively, while
$$(c_1, \{x, y, z\}),\, (c_2, \{x, z\}),\, (c_2, \{y\}\cup P'),\, (c_3, \{x, y, z\}\cup Q')$$
are vertices of ${\rm Int}(\gamma')$. The partition intersection graphs of $\Pi(\gamma)$ and $\Pi(\gamma')$ are illustrated in Figure~\ref{intersection1}. It is now easily checked that ${\rm Int}(\gamma')$ can be constructed from ${\rm Int}(\gamma)$ by relabelling $(c_1, \{x, y\})$ as $(c_1, \{x, y, z\})$, $(c_2, \{x, y\}\cup P')$ as $(c_2, \{y\}\cup P')$, and $(c_3, \{x, y\}\cup Q')$ as $(c_3, \{x, y, z\}\cup Q')$, and adding a new vertex $(c_2, \{x, z\})$ adjacent to precisely $(c_1, \{x, y, z\})$ and $(c_3, \{x, y, z\}\cup Q')$.

\begin{figure}
\center
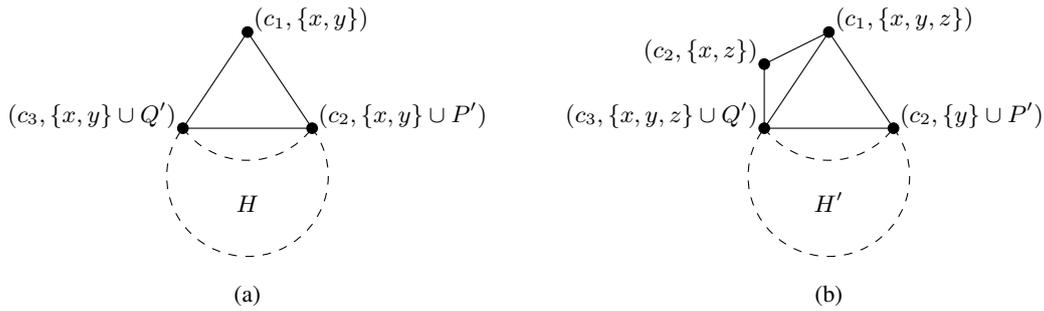
\caption{Illustrations of (a) ${\rm Int}(\gamma)$ and (b) ${\rm Int}(\gamma')$ in the proof of Lemma~\ref{fork}, where $H$ and $H'$ represent the parts of ${\rm Int}(\gamma)$ and ${\rm Int}(\gamma')$ not explicitly shown.}
\label{intersection1}
\end{figure}

Since $\Pi(\gamma')$ is convex on $T'$, the partition intersection graph ${\rm Int}(\gamma')$ has a restricted chordal completion \blue{by Theorem~\ref{thm:key-help1}}. If $G'$ is a minimal restricted chordal completion of ${\rm Int}(\gamma')$, then, by Lemma~\ref{completion}, no completion edge of $G'$ is incident with $(c_2, \{x, z\})$. Therefore if ${\rm Int}(\gamma')$ has two distinct minimal restricted chordal completions, then, by the above construction, ${\rm Int}(\gamma)$ has two distinct minimal restricted chordal completions, a contradiction as $\Pi(\gamma)$ defines $T$. Thus ${\rm Int}(\gamma')$ has a unique minimal restricted chordal completion. Since $\Pi(\gamma')$ distinguishes $T'$, it follows that $\gamma'$ defines $T'$, and so $T'$ is defined by three characters. The proof that if $T'$ is defined by three characters, then $T$ is defined by three characters is the same argument but in reverse. This completes the proof of the lemma.
\end{proof}

In contrast to a fork modification, a balanced modification does not necessarily preserve the property of being defined by a set of at most three characters. However, it does preserve the property of not being defined by a set of at most three characters.

\begin{lemma}
Let $T$ be a binary phylogenetic $X$-tree, where $|X|\ge 3$, and let $T'$ be a binary phylogenetic tree obtained from $T$ by a balanced modification. If $T$ is not defined by a set of at most three characters, then $T'$ is not defined by a set of at most three characters.
\label{balanced}
\end{lemma}

\begin{proof}
Suppose that $T$ is not defined by at most three characters. Then $T$ is not a caterpillar and so $T$, and therefore $T'$, has an internal vertex incident with three internal edges and $|X|\ge 6$. Let $(x, y)$ be a cherry of $T$. We may assume that $T'$ is obtained from $T$ as shown in Figure~\ref{modifications}(b). Thus $(x, z)$ and $(w, y)$ are cherries of $T'$ and the leaf set of $T'$ is $X\cup \{w, z\}$ and, as $|X|\ge 6$, either $|P|\ge 2$ or $|Q|\ge 2$. There are two cases: (i) exactly one of $P$ and $Q$ has size one and (ii) $|P|, |Q|\ge 2$. We will establish the lemma for (ii). The proof of (i) is similar and omitted.

Suppose that $|P|, |Q|\ge 2$ and, \blue{for the purpose of obtaining a contradiction,} that $T'$ is defined by three characters. By Proposition~\ref{cla:internal}, $T'$ is defined by an internal $3$-colouring $\gamma'$. Let $\{c_1, c_2, c_3\}$ be the codomain of $\gamma'$ and let $\gamma$ be the restriction of $\gamma'$ to the edges of $T$. That is, \blue{for each $e\in E^0(T)$, we have} $\gamma(e)=\gamma'(e)$. Since $|P|, |Q|\ge 2$, we may assume without loss of generality that
$$(c_1, \{w, x, y, z\}),\, (c_2, \{x, z\}),\, (c_2, \{w, y\}\cup P'),\, (c_3, \{w, y\}),\, (c_3, \{x, z\}\cup Q')$$
are vertices of ${\rm Int}(\gamma')$, where $P'$ and $Q'$ are non-empty subsets of $P$ and $Q$, respectively, while
\blue{$$(c_1, \{x, y\}),\, (c_2, \{x, y\}\cup P'),\, (c_3, \{x, y\}\cup Q')$$}
are vertices of ${\rm Int}(\gamma)$. The partition intersection graphs of $\Pi(\gamma')$ and $\Pi(\gamma)$ are illustrated in Figure~\ref{intersection2}. A routine check shows that ${\rm Int}(\gamma)$ can be constructed from ${\rm Int}(\gamma')$ by deleting $(c_2, \{x, z\})$ and $(c_3, \{w, y\})$, relabelling $(c_2, \{w, y\}\cup P')$ as \blue{$(c_2, \{x, y\}\cup P')$} and $(c_3, \{x, z\}\cup Q')$ as \blue{$(c_3, \{x, y\}\cup Q')$}, and joining \blue{$(c_2, \{x, y\}\cup P')$} and \blue{$(c_3, \{x, y\}\cup Q')$} with an edge.

\begin{figure}
\center
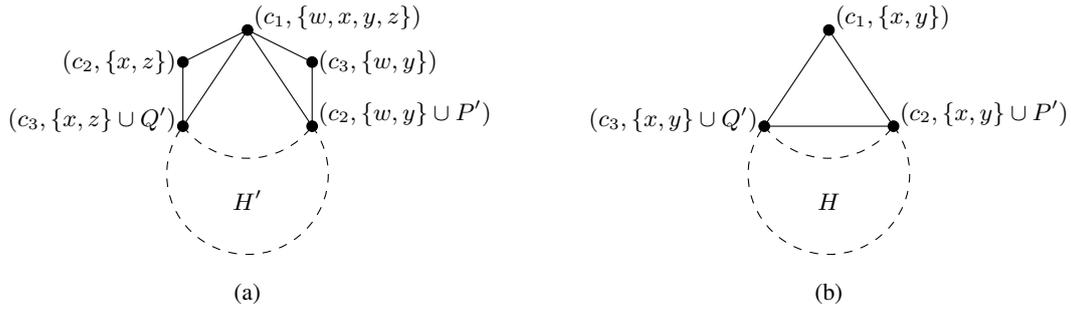
\caption{Illustrations of (a) ${\rm Int}(\gamma')$ and (b) ${\rm Int}(\gamma)$ in the proof of Lemma~\ref{balanced}, where $H'$ and $H$ represent the parts of ${\rm Int}(\gamma')$ and ${\rm Int}(\gamma)$ not explicitly shown.}
\label{intersection2}
\end{figure}

Since $\gamma'$ defines $T'$, the partition intersection graph ${\rm Int}(\gamma')$ has a unique minimal restricted chordal completion $G'$. As ${\rm Int}(\gamma')$ has no vertex-induced cycles of size at least four containing either $(c_2, \{x, z\})$ or $(c_3, \{w, y\})$, it follows by Lemma~\ref{completion} that no completion edge of $G'$ is incident with these vertices. Furthermore,
$$\{(c_2, \{w, y\}\cup P'), (c_3, \{x, z\}\cup Q')\}$$
is a completion edge of $G'$. To see this, let $G$ be a minimal restricted chordal completion of ${\rm Int}(\gamma)$. \blue{Since $\gamma$ is convex on $T$, such a graph exists.} Let $E_1$ denote the set of completion edges of $G$ and let $E'_1$ denote the collection of $2$-element subsets of $V({\rm Int}(\gamma'))$ obtained from $E_1$ by replacing $(c_2, \{x, y\}\cup P')$ with $(c_2, \{w, y\}\cup P')$ and $(c_3, \{x, y\}\cup Q')$ with $(c_3, \{x, z\}\cup Q')$. Since $G$ is a minimal restricted chordal completion of ${\rm Int}(\gamma)$, it follows by the above construction that the graph obtained from ${\rm Int}(\gamma')$ by adding the edges in $E'_1$ as well as the edge $\{(c_2, \{w, y\}\cup P'), (c_3, \{x, z\}\cup Q')\}$ to ${\rm Int}(\gamma')$ is a minimal restricted chordal completion $G'_1$ of ${\rm Int}(\gamma')$. If $\{(c_2, \{w, y\}\cup P'), (c_3, \{x, z\}\cup Q')\}$ is not a completion edge of $G'$, then, as this edge is a completion edge of $G'_1$, it follows that ${\rm Int}(\gamma')$ has two distinct minimal restricted chordal completions. This is a contradiction as $\gamma'$ defines $T'$. It now follows that, as ${\rm Int}(\gamma')$ has a unique minimal restricted chordal completion, ${\rm Int}(\gamma)$ has a unique minimal restricted chordal completion. Since $\Pi(\gamma)$ is convex on $T$ and distinguishes $T$, we deduce that $\gamma$ defines $T$, and so $T$ is defined by three characters. This last contradiction completes the proof of the lemma.
\end{proof}

The next theorem is the necessary direction of Theorem~\ref{thm:main}.

\begin{theorem}
Let $T$ be a binary phylogenetic $X$-tree, and suppose that $T$ has an internal subtree isomorphic to the snowflake. Then, up to isomorphism, $T$ can be obtained from the phylogenetic tree in Figure~\ref{snowflake}(a) by a sequence of cherry modifications. In particular, $T$ is not defined by a set of at most three characters.
\label{sequence1}
\end{theorem}

\begin{proof}
Since $T$ has an internal subtree isomorphic to the snowflake, $|X|\ge 12$. The proof is by induction on $n=|X|$. If $n=12$, then $T$ is isomorphic to the phylogenetic tree shown in Figure~\ref{snowflake}(a) and so, by Lemma~\ref{obs:twelve}, the \blue{theorem} holds.

Now suppose that the \blue{theorem} holds for all binary phylogenetic trees with at most $n-1\ge 12$ leaves. Let $S$ be an internal subtree of $T$ isomorphic to the snowflake, and let $v$ be the unique vertex of $S$ that is at distance two from each of its leaves. Now let $u_1$ be an internal vertex of $T$ at maximum distance from $v$. Observe that $u_1$ is adjacent to two leaves, say $x$ and $z$. Let $u$ denote the unique internal vertex of $T$ adjacent to $u_1$, and let $u_2$ denote the vertex of $T$ adjacent to $u$, that is not $u_1$, but at the same distance from $v$ as $u_1$. If $u_2$ is a leaf, label this vertex $y$. Otherwise, $u_2$ is adjacent to two leaves, say $y$ and $w$.

Let $T'$ be the binary phylogenetic tree such that $T$ is obtained from $T'$ by \blue{applying} either a fork modification to the cherry $(x, y)$ or a balanced modification to the cherry $(x, y)$. Since $n\ge 13$, it follows by the choice of $u_1$ that neither $u_1$ nor $u_2$ are vertices in $S$. Thus $T'$ has an internal subtree isomorphic to the snowflake. Therefore, by induction, $T'$ can be obtained from the binary phylogenetic tree shown in Figure~\ref{snowflake}(a) by a sequence of cherry modifications and, moreover, $T'$ is not defined by three characters. Theorem~\ref{sequence1} now follows by Lemmas~\ref{fork} and~\ref{balanced}.
\end{proof}

\section{Proof of Theorem~\ref{thm:main}: Sufficient Direction}
\label{sufficient}

In this section, we complete the proof of Theorem~\ref{thm:main} by proving the sufficient direction. To do this, we first introduce the operation of cherry union. Let $T_1$ and $T_2$ be two binary phylogenetic trees with leaf sets $X_1$ and $X_2$, respectively, such that $|X_1|, |X_2|\ge 4$ and $X_1\cap X_2=\{x\}$. In addition, suppose that $T_1$ has a cherry $(x, y_1)$ and $T_2$ has a cherry $(x, y_2)$. Let $T$ be the binary phylogenetic tree with leaf set $(X_1\cup X_2)-\{y_1, y_2\}$ that is obtained from $T_1$ and $T_2$ by deleting the leaf $x$ in exactly one of $T_1$ and $T_2$, identifying the vertices $y_1$ and $y_2$, and suppressing the two resulting degree-$2$ vertices. We say that $T$ is the {\em cherry union} of $T_1$ and $T_2$, and denote $T$ by $T_1\Box T_2$. An illustration of a cherry union is shown in Figure~\ref{amalgam1}.

\begin{figure}
\center
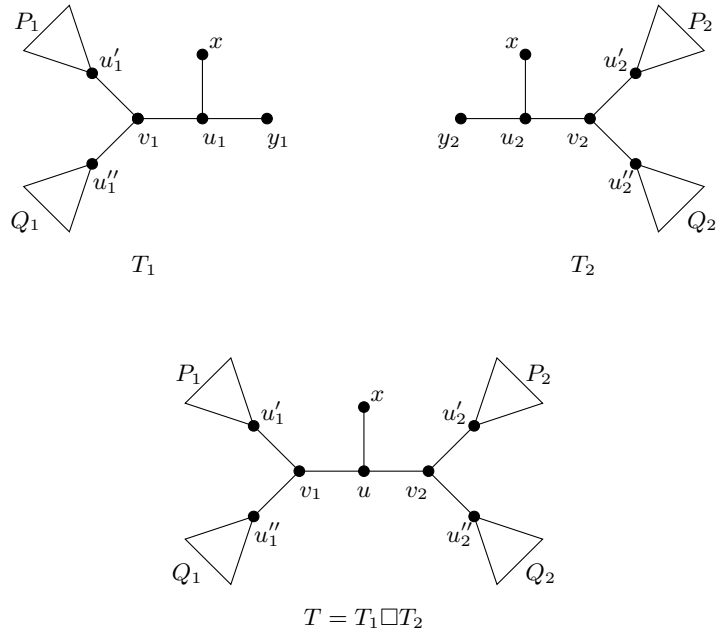
\caption{The cherry union $T$ of two binary phylogenetic trees $T_1$ and $T_2$, where the leaf set of $T_1$ is $P_1\cup Q_1\cup\{x, y_1\}$, the leaf set of $T_2$ is $P_2\cup Q_2\cup \{x, y_2\}$, and the leaf set of $T$ is $P_1\cup Q_1\cup P_2\cup Q_2\cup \{x\}$.}
\label{amalgam1}
\end{figure}

The next lemma shows that a cherry union preserves the property of being defined by a set of at most three characters. As in the last section, if $\gamma$ is an internal $k$-colouring of a binary phylogenetic tree, we write the vertices of ${\rm Int}(\gamma)$ in the form $(c, A)$ instead of $(\pi(\gamma, c), A)$.

\begin{lemma}
Let $T_1$ and $T_2$ be two binary phylogenetic trees with leaf sets $X_1$ and $X_2$, respectively, where $|X_1|, |X_2|\ge 4$. Suppose that $X_1\cap X_2=\{x\}$, and $(x, y_1)$ and $(x, y_2)$ are cherries in $T_1$ and $T_2$, respectively. If $T_1$ and $T_2$ are each defined by a set of at most three characters, then $T_1\Box T_2$ is defined by a set of at most three characters.
\label{amalgam2}
\end{lemma}

\begin{proof}
We may assume that $T_1$, $T_2$, and $T_1\Box T_2$ are as shown in Figure~\ref{amalgam1}. Suppose that $T_1$ and $T_2$ are each defined by a set of at most three characters. If one of $T_1$ and $T_2$, say $T_2$ is a caterpillar, then it is easily seen that, up to leaf labels, $T_1\Box T_2$ is obtained from $T_1$ by a sequence of fork modifications and so, by Lemma~\ref{fork}, $T_1\Box T_2$ is defined by most three characters. Thus we may assume that neither $T_1$ nor $T_2$ is a caterpillar. By Proposition~\ref{cla:internal}, $T_1$ and $T_2$ are defined by internal $3$-colourings $\gamma_1$ and $\gamma_2$, respectively. Without loss of generality, we may assume that the codomain of $\gamma_1$ and the codomain of $\gamma_2$ is $\{c_1, c_2, c_3\}$. Moreover, by recolouring if necessary, we may assume that if $u_1$ and $u_2$ denote the vertices of $T_1$ and $T_2$ adjacent to $x$ and $y_1$, and adjacent to $x$ and $y_2$, respectively, then the (unique) internal edges of $T_1$ and $T_2$ incident with $u_1$ and $u_2$ are assigned different colours.

Since neither $T_1$ nor $T_2$ is a caterpillar, $|P_1|, |P_2|\ge 2$. There are three cases to consider: (i) $|Q_1|=1=|Q_2|$, (ii) exactly one of $Q_1$ and $Q_2$ has size one, and (iii) $|Q_1|, |Q_2|\ge 2$. We establish the lemma for (iii). The proofs of (i) and (ii) are similar and omitted.

Suppose that (iii) holds. Then, without loss of generality, we may assume that $(c_1, \{x, y_1\})$, $(c_2, \{x, y_1\}\cup P'_1)$, and $(c_3, \{x, y_1\}\cup Q'_1)$ are vertices of ${\rm Int}(\gamma_1)$, where $P'_1$ and $Q'_1$ are non-empty subsets of $P_1$ and $Q_1$, respectively. Similarly, $(c_2, \{x, y_2\})$, $(c_1, \{x, y_2\}\cup P'_2)$, and $(c_3, \{x, y_2\}\cup Q'_2)$ are vertices of ${\rm Int}(\gamma_2)$, where $P'_2$ and $Q'_2$ are non-empty subsets of $P_2$ and $Q_2$, respectively. Illustrations of ${\rm Int}(\gamma_1)$ and ${\rm Int}(\gamma_2)$ are shown in Figure~\ref{intersection3}. Let $\gamma$ be the internal $3$-colouring of $T=T_1\Box T_2$ induced by $\gamma_1$ and $\gamma_2$. Thus, in reference to Figure~\ref{amalgam1}, $\{u, v_1\}$ and $\{u''_2, v_2\}$ are coloured $c_1$, $\{u''_1, v_1\}$ and $\{u, v_2\}$ are coloured $c_2$, and $\{u'_1, v_1\}$ and $\{u'_2, v_2\}$ are coloured $c_3$. Since $\gamma$ is an internal $3$-colouring of $T$, it follows that $\Pi(\gamma)$ is convex on $T$ and distinguishes $T$.

\begin{figure}
\center
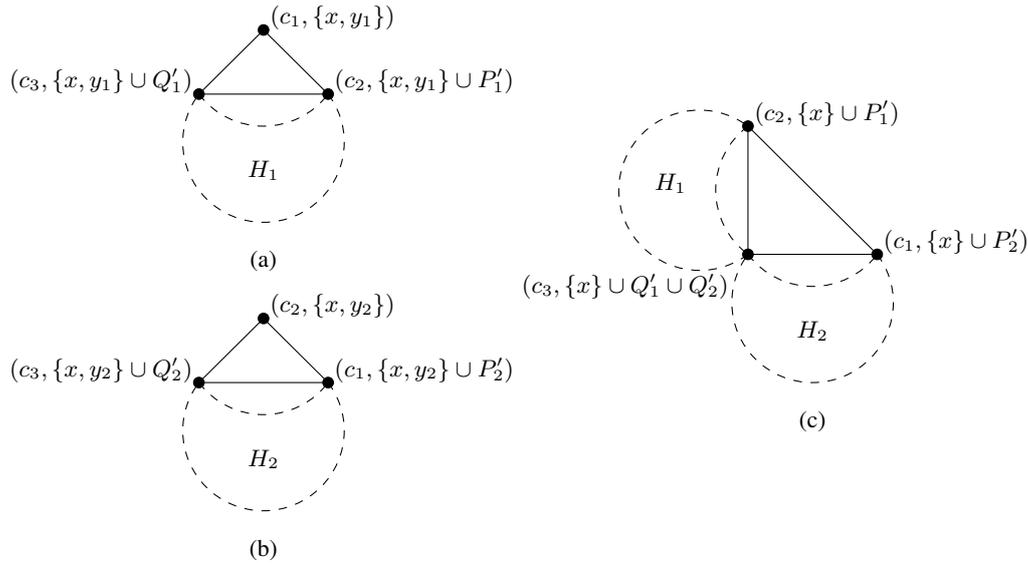
\caption{Illustrations of (a) ${\rm Int}(\gamma_1)$, (b) ${\rm Int}(\gamma_2)$, and (c) ${\rm Int}(\gamma)$ in the proof of Lemma~\ref{amalgam2}, where $H_1$ and $H_2$ represent the parts of ${\rm Int}(\gamma_1)$ and ${\rm Int}(\gamma_2)$ not explicitly shown.}
\label{intersection3}
\end{figure}

Now consider ${\rm Int}(\gamma)$, and observe that $(c_1, \{x\}\cup P'_2)$, $(c_2, \{x\}\cup P'_1)$, and $(c_3, \{x\}\cup Q'_1\cup Q'_2)$ are vertices of ${\rm Int}(\gamma)$. Also, as $\Pi(\gamma)$ is convex on $T$, we have in addition to these three vertices, for all $i\in \{1, 2, 3\}$, that $(c_i, D)$ is a vertex of ${\rm Int}(\gamma)$ if and only if either $(c_i, D)$ is a vertex of ${\rm Int}(\gamma_1)$ or $(c_i, D)$ is a vertex of ${\rm Int}(\gamma_2)$. It is now easily checked that ${\rm Int}(\gamma)$ can be constructed from ${\rm Int}(\gamma_1)$ and ${\rm Int}(\gamma_2)$ by identifying the vertices $(c_1, \{x, y_1\})$ and $(c_1, \{x, y_2\}\cup P'_2)$, $(c_2, \{x, y_1\}\cup P'_1)$ and $(c_2, \{x, y_2\})$, and $(c_3, \{x, y_1\}\cup Q'_1)$ and $(c_3, \{x, y_2\}\cup Q'_2)$ together with the corresponding edges, and then relabelling the identified vertices as $(c_1, \{x\}\cup P'_2)$, $(c_2, \{x\}\cup P'_1)$, and $(c_3, \{x\}\cup Q'_1\cup Q'_2)$, respectively. An illustration of ${\rm Int}(\gamma)$ is shown in Figure~\ref{intersection3}.

Let $G$ be a minimal restricted chordal completion of ${\rm Int}(\gamma)$. Since $(c_1, \{x\}\cup P'_2)$, $(c_2, \{x\}\cup P'_1)$, and $(c_3, \{x\}\cup Q'_1\cup Q'_2)$ is a $3$-clique of ${\rm Int}(\gamma)$, it follows by the above construction that if $C$ is a vertex-induced cycle of ${\rm Int}(\gamma)$ with at least four vertices, then, modulo replacing $(c_2, \{x\}\cup P'_1)$ with $(c_2, \{x, y_1\}\cup P'_1)$ and $(c_3, \{x\}\cup Q'_1\cup Q'_2)$ with $(c_3, \{x, y_1\}\cup Q'_1)$, or \blue{$(c_1, \{x\}\cup P'_2)$} with $(c_1, \{x, y_2\}\cup P'_2)$ and $(c_3, \{x\}\cup Q'_1\cup Q'_2)$ with $(c_3, \{x, y_2\}\cup Q'_2)$, the cycle $C$ is a vertex-induced cycle of either ${\rm Int}(\gamma_1)$ or ${\rm Int}(\gamma_2)$ with at least four vertices. Thus, by Lemma~\ref{completion} we deduce that if ${\rm Int}(\gamma)$ has two minimal restricted chordal completions, then either ${\rm Int}(\gamma_1)$ or ${\rm Int}(\gamma_2)$ has two minimal restricted chordal completions, contradicting that $\gamma_1$ and $\gamma_2$ define $T_1$ and $T_2$, respectively. Hence $G$ is the unique minimal restricted chordal completion of ${\rm Int}(\gamma)$. Since $\gamma$ is convex on $T$ and distinguishes $T$, the internal $3$-colouring $\gamma$ defines $T$. This completes the proof of the lemma.
\end{proof}

A binary phylogenetic $X$-tree $T$ is a {\em cherried caterpillar} if either $|X|=4$ or $T$ can be obtained from a caterpillar by replacing each leaf with a pair of leaves in a cherry, that is, for each pendant edge of a caterpillar, subdividing it and adjoining a new leaf by adding an edge joining the new leaf and the subdivision vertex. A cherried caterpillar is illustrated in Figure~\ref{orchard1}.

\begin{figure}
\center
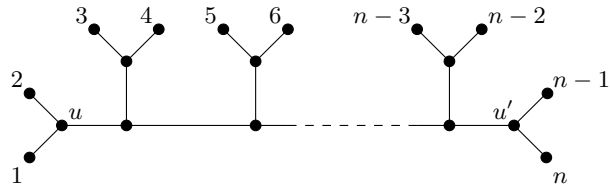
\caption{A cherried caterpillar with leaf set $[n]$, where $n\ge 4$.}
\label{orchard1}
\end{figure}

\begin{lemma}
Let $T$ be a cherried caterpillar. Then $T$ is defined by a set of at most three characters.
\label{orchard2}
\end{lemma}

\begin{proof}
Without loss of generality, we may assume that the leaf set of $T$ is $[n]$ and that its leaves are labelled as shown in Figure~\ref{orchard1}. If $n=4$, then the lemma trivially holds. So assume that $|X|\ge 6$. Let $u$ and $u'$ be the internal vertices of $T$ adjacent to the leaves $1$ and $2$, and adjacent to the leaves $n-1$ and $n$, respectively. Now consider an internal $3$-colouring $\gamma$ of $T$ that assigns each edge on the path joining $u$ and $u'$ one of two colours, say $c_1$ and $c_2$, and assigns all remaining internal edges the third colour, say $c_3$. Let $\{\chi_1, \chi_2, \chi_3\}$ denote the set of characters on $[n]$ induced by $\gamma$, where $\chi_i$ is the character induced by $c_i$ for all $i\in \{1, 2, 3\}$. If $|X|=6$, then ${\rm Int}(\C)$ is chordal, and so it has a unique minimal restricted chordal completion. Furthermore, if $|X|\ge 8$, then it is easily checked that ${\rm Int}(\C)$ has a unique vertex-induced cycle $C$ with at least four vertices. Moreover, $C$ has a unique vertex whose first coordinate is $\chi_3$. The second coordinate of this vertex is $\{1, 2, n-1, n\}$.

Let $G$ be a minimal restricted chordal completion of ${\rm Int}(\C)$. If there is a completion edge of $G$ joining a vertex in $C$ whose first coordinate is $\chi_1$ and a vertex in $C$ whose first coordinate is $\chi_2$, then $G$ contains a cycle $C'$ with at least four vertices all of which are vertices of $C$ whose first coordinates are either $\chi_1$ or $\chi_2$. But by considering the subgraph of $G$ induced by the vertices of $C'$, it is easily checked that this subgraph is not chordal, a contradiction. Therefore all of the completion edges of $G$ are incident with the unique vertex in $C$ whose first coordinate is $\chi_3$. It now follows that ${\rm Int}(\gamma)$ has a unique minimal restricted chordal completion. Since $\C$ is convex on $T$ and $\C$ distinguishes $T$, the lemma follows by Theorem~\ref{thm:key-help}.
\end{proof}

\begin{lemma}
Let $T$ be a binary phylogenetic $X$-tree, where $|X|\ge 4$, and suppose that $T$ has no internal subtree isomorphic to the snowflake. Then either $T$ is a cherried caterpillar or $T$ has a leaf not in a cherry.
\label{non-cherry-leaf}
\end{lemma}

\begin{proof}
If $T$ has a leaf not in a cherry, then the lemma holds, so suppose that every leaf of $T$ is in a cherry. If $|X|=4$, then $T$ is a cherried caterpillar, so we may assume that $|X|\ge 6$. Let $u$ and $u'$ be internal vertices of $T$ such that, amongst all pairs of internal vertices, the length of the path $P$ connecting $u$ and $u'$ is maximised. Let $P=u, v_1, v_2, \ldots, v_k, u'$, where $k\ge 1$ as $|X|\ge 6$. By maximality, $u$ is adjacent to two leaves and $u'$ is adjacent to two leaves. For all $i\in \{1, 2, \ldots, k\}$, let $w_i$ denote the vertex of $T$ adjacent to $v_i$ that is not on $P$. Since every leaf of $T$ is in a cherry, $w_i$ is not a leaf, and so $w_i$ is an internal vertex of $T$ for all $i$. Consider $w_1$. Since the length of the path connecting $w_1$ and $u'$ is the same as the length of $P$, it follows by the maximality of $P$ that $w_1$ adjacent to two leaves. Similarly, $w_k$ is also adjacent to two leaves. If $k\in \{1, 2\}$, this implies that $T$ is a cherried caterpillar. So assume that $k\ge 3$. Now consider $w_j$, where $j\in \{2, 3, \ldots, k-1\}$. If $w_j$ is not adjacent to two leaves, then, as every leaf of $T$ is in a cherry, $w_j$ is adjacent to two internal vertices $w'_j$ and $w''_j$, neither of which is $v_j$. But then the internal subtree of $T$ induced by the vertices in
$$\{w'_j, w''_j, w_j, v_j, v_{j-1}, w_{j-1}, v_{j-2}, v_{j+1}, w_{j+1}, v_{j+2}\},$$
where $v_{j-2}=u$ if $j=2$ and $v_{j+2}=u'$ if $j=k-1$, is isomorphic to the snowflake, a contradiction. Thus, for all $j$, we have that $w_j$ is adjacent to two leaves. In particular, $T$ is a cherried caterpillar. This completes the proof of the lemma.
\end{proof}

The next theorem proves the sufficient direction of Theorem~\ref{thm:main}, thereby completing its proof.

\begin{theorem}
Let $T$ be a binary phylogenetic $X$-tree, where $|X|\ge 4$, and suppose that $T$ has no internal subtree isomorphic to the snowflake. Then $T$ is defined by a set of at most three characters.
\label{3-characters}
\end{theorem}

\begin{proof}
The proof is by induction on $n=|X|$. If $n=4$, the \blue{theorem} trivially holds. Suppose that $n\ge 5$, and that the \blue{theorem} holds for all binary phylogenetic trees whose leaf sets have size at most $n-1\ge 4$. If $T$ is a cherried caterpillar, then, by Lemma~\ref{orchard2}, $T$ is defined by a set of at most three characters. Therefore assume that $T$ is not a cherried caterpillar. Then, by Lemma~\ref{non-cherry-leaf}, $T$ has a leaf, $x$ say, not in a cherry. Let $v$ be the internal vertex of $T$ adjacent to $x$, and let $u_1$ and $u_2$ be the internal vertices of $T$ adjacent to $v$. Let $F_1$ denote the forest obtained from $T$ by deleting the two edges incident with $u_2$ that are not $\{v, u_2\}$, and let $T_1$ denote the binary phylogenetic $X_1$-tree obtained from the component of $F_1$ containing $u_1$ by relabelling $u_2$, now a leaf, as $z_1$, where $z_1\not\in X$. Similarly, let $F_2$ denote the forest obtained from $T$ by deleting the two edges incident with $u_1$ that are not $\{v, u_1\}$, and let $T_2$ denote the binary phylogenetic $X_2$-tree obtained from the component of $F_2$ containing $u_2$ by relabelling $u_1$, now a leaf, as $z_2$, where $z_2\not\in X\cup \{z_1\}$. Observe that $4\le |X_1|, |X_2|\le n-1$, $X_1\cap X_2=\{x\}$ and, more particularly, $T$ \blue{is isomorphic to} $T_1\Box T_2$.

Since $T$ has no internal subtree isomorphic to the snowflake, it is easily checked that $T_1$ has no internal subtree isomorphic to the snowflake. Thus, by the induction assumption, $T_1$ is defined by three characters. Similarly, $T_2$ is also defined by three characters. As $T$ \blue{is isomorphic to} $T_1\Box T_2$, Theorem~\ref{3-characters} now follows by Lemma~\ref{amalgam2}.
\end{proof}

\section{Discussion}
\label{sect:discuss}

Since each binary phylogenetic tree is defined by at most four characters~\cite{huber2005four}, the results presented in this paper partitions the set of all binary phylogenetic $X$-trees, where \blue{$|X|\ge 5$}, into three classes. In particular, those \blue{binary phylogenetic $X$-trees} that are defined by \blue{a set of} exactly two, those defined by \blue{a set of} three but \blue{no less than three}, and those defined by \blue{a set of} four but \blue{no less than four} characters. Moreover, given a \blue{binary phylogenetic $X$-tree} $T$, we can decide which class $T$ is contained in \blue{in time that is linear} in $|X|$. To see this, let $\mathcal C_2(X)$, $\mathcal C_3(X)$, and $\mathcal C_4(X)$ denote these three classes of binary phylogenetic trees, respectively, and note that a binary phylogenetic $X$-tree has \blue{$2|X|-2$} vertices \cite[Proposition 2.1.3]{semple2003phylogenetics}. Now, as $T$ is contained in $\mathcal C_2(X)$ if and only if $T$ is a caterpillar tree by Theorem~\ref{caterpillar2}, it is clear that containment in \blue{$\mathcal C_2(X)$} can be checked in $O(|X|)$ time. Furthermore, if $T$ is not a caterpillar, then we can again check in $O(|X|)$ time if $T$ is contained in $\mathcal C_3(X)$ or $\mathcal C_4(X)$ using Theorem~\ref{thm:main} by simply considering, \blue{for} each internal vertex $u$, the internal vertices \blue{of $T$ at distance at most} two from $u$, and checking whether or not they induce a snowflake.

There remain some interesting questions and investigations. For example, given a binary phylogenetic $X$-tree $T$ in $\mathcal C_3(X)$ or $\mathcal C_4(X)$, can we 
determine (up to the natural notion of equivalence) the number of ways that $T$ can be defined by three or four characters, respectively? Also, it would be interesting to compute
$$\lim_{|X|\to \infty} \frac{|\mathcal C_2(X) \cup \mathcal C_3(X)|}{|\mathcal C_4(X)|}.$$
This could be of practical interest since, if this limit is $0$, it would imply that, as the size of $X$ grows, almost all binary phylogenetic $X$-trees are \blue{in $\C_4(X)$}.

Finally, the notion of defining a binary phylogenetic tree can be generalised to the weaker notion of ``identifiability" \cite{BHS05}. In particular, an {\em $X$-tree} is an ordered pair $(T; \phi)$ consisting of a tree $T$ with vertex set, $V$ say, and a map $\phi:X\rightarrow V$ such that if $v\in V$ has degree at most two, then $v\in \phi(X)$. A phylogenetic $X$-tree is an $X$-tree in which $\phi$ is a bijection from $X$ to the set of leaves of $T$. Intuitively, an $X$-tree can be obtained by contracting some edges of a binary phylogenetic $X$-tree \cite{semple2003phylogenetics}. A collection of characters on $X$ {\em identifies} an $X$-tree $(T; \phi)$ if \blue{$\C$ is convex on $T$} (analogous to that described in this paper) and all other $X$-trees \blue{on which $\mathcal C$ is convex} are ``refinements" of $T$.

In \cite{bordewich2015defining}, it is proven that if $d$ is the maximum degree of any vertex in an $X$-tree $(T; \phi)$ and $k$ is a positive integer, then, in case $$k=4\lceil \log_2(d-2)\rceil+4,$$
there is a collection of $k$ characters that identifies $(T; \phi)$ and, in case $k < \log_2 d$, there is no collection of $k$ characters that identifies $T$. Bearing this in mind, it would be interesting to investigate whether the results in \blue{the present} paper for binary phylogenetic trees can be extended in some way to $X$-trees.

\section*{Acknowledgements}

KTH and VM would like to thank \blue{The Royal Society in the context of its International Exchanges Scheme} for support and also the University of Canterbury for hosting them during a brief visit. SL and CS were supported by the New Zealand Marsden Fund. All authors would like to thank the  Institute for Mathematical Sciences, National University of Singapore where some of the research was partially completed while they were visiting in 2023.


\end{document}